\documentclass[a4paper,reqno]{amsart}

\usepackage[centertags]{amsmath}
\usepackage{amsfonts}
\usepackage{euscript}
\usepackage{amssymb}
\usepackage{amsthm}
\usepackage{newlfont}
\usepackage{mathrsfs}
\usepackage{euscript}
\usepackage{graphicx}
\usepackage{caption}
\usepackage[usenames,dvipsnames]{color}
\usepackage{enumerate}

\theoremstyle{plain}
  \newtheorem{theorem}{Theorem}[section]
  
  \newtheorem{proposition}[theorem]{Proposition}
  
\theoremstyle{definition}

\theoremstyle{remark}

\numberwithin{equation}{section}

\newcommand{\caA}{{\mathcal A}}

\newcommand{\caC}{{\mathcal C}}

\newcommand{\caH}{{\mathcal H}}
\newcommand{\caI}{{\mathcal I}}

\newcommand{\caP}{{\mathcal P}}

\newcommand{\bbC}{{\mathbb C}}

\newcommand{\bbR}{{\mathbb R}}

\newcommand{\bbZ}{{\mathbb Z}}

\DeclareMathAlphabet{\mathpzc}{OT1}{pzc}{m}{it}

\DeclareMathOperator{\tr}{tr}

\graphicspath{{Drawings/}}

\newcommand{\beq}{\begin{equation}}
\newcommand{\eeq}{\end{equation}}
\newcommand{\baq}{\begin{eqnarray}}
\newcommand{\eaq}{\end{eqnarray}}

\renewcommand{\d}{\mathrm{d}}
\newcommand{\e}{\mathrm{e}}

\newcommand{\opunit}{\text{1}\kern-0.22em\text{l}}

\title{Local Perturbations Perturb --Exponentially-- Locally }

\author{W.~De Roeck}
\address[W.~De Roeck]{Instituut voor Theoretische Fysica, K.~U.~Leuven, Celestijnenlaan 200D, B-3001 Heverlee (Belgium)}
\email{wojciech.deroeck@fys.kuleuven.be }
\author{M.~Sch\"utz}
\address[M.~Sch\"utz]{Instituut voor Theoretische Fysica, K.~U.~Leuven, Celestijnenlaan 200D, B-3001 Heverlee (Belgium)}
\email{marius.schutz@fys.kuleuven.be}

\begin{document}

\begin{abstract}
We elaborate on the principle that for gapped quantum spin systems with local interaction ``local perturbations [in the Hamiltonian] perturb  locally [the ground state]''.  This principle was established in \cite{BMNS}, relying on the `spectral flow technique' or `quasi-adiabatic continuation' \cite{H}  to obtain locality estimates with sub-exponential decay in the distance to the spatial support of the perturbation. 
We use  ideas of \cite{HMNS} to obtain similarly a transformation between gapped eigenvectors and their perturbations that is local with exponential decay.  This allows to improve locality bounds on the effect of perturbations on the low lying states in certain gapped models with a unique `bulk ground state' or `topological quantum order'. We also give some estimate on the exponential decay of correlations in models with impurities  where some relevant correlations decay faster than one would naively infer from the global gap of the system, as one also expects in disordered systems with a localized ground state.  \\
\end{abstract}

\maketitle

\section{Introduction}

This article concerns quantum spin systems with local interactions, whose dynamics typically satisfy a Lieb--Robinson propagation bound \cite{LR}.
We consider paths of Hamiltonians that are modified only locally in a confined spatial region and that maintain a gapped sector in the spectrum. We ask to what (spatial) extent the eigenvectors from the gapped sector are affected by such modifications. This setting in particular includes the study of ground states. The regions where the Hamiltonian is modified, may be viewed as impurities.

More specifically, we investigate the principle that ``local perturbations perturb locally''(LPPL), which was formulated in \cite{BMNS}. The authors of \cite{BMNS} give a general and rigorous account of the `spectral flow technique' (also named `quasi-adiabatic continuation') introduced by Hastings in seminal work \cite{H}, see also \cite{HW}, which allows to relate the spectral projections of an isolated patch of spectrum along the Hamiltonian path by means of a quasi-local unitary flow. The spectral flow is quasi-local in the sense that the involved unitary operator can be approximated by a truly local one up to an error that decreases (only) `sub-exponentially' in the diameter of the support of this local operator. As an application, they show the principle LPPL to hold in the sense that for a unique ground state the effect of the perturbation' on the expectation value of local observables decays (only) sub-exponentially in the distance.

Instead of using the spectral flow technique, we relate the individual eigenvectors from the gapped sector (along the path) by transformations that are indeed exponentially local by building on ideas from \cite{HMNS}, based on Lieb--Robinson bounds.
These transformations are the content of our main Theorem \ref{thm: main} in Section \ref{sec: results}. As a drawback, these transformations do not correspond to a unitary linear map between the  spectral subspaces (In fact, they do not even correspond to a linear map, unless one gives up exponential locality).  Moreover, they do not allow as such to consider perturbations in the whole volume, in contrast to the spectral flow (in that case however, it is usually hard to verify that the spectral gap condition remains satisfied along the path, see e.g.~\cite{MZ}). Here we are exclusively concerned with perturbations that have bounded support and hence norm-bounds independent of the volume.\\

By invoking the exponential clustering property for gapped ground states (see Theorem \ref{thm: ExpCluster}) our main result enables us to improve the LPPL principle to exponential precision for systems with unique ground state (also covered already in \cite{HMNS}) or with topological quantum order, thus justifying our title. 
Furthermore, we may also dope these systems with impurities that increase the dimension of the low lying sector by locally closing the gap. This is modeled by adding on certain sites additional degrees of freedom with degenerate energies. In that setup, we arrive at statements on individual eigenstates from that sector without making further assumptions on the spectrum in that sector (e.g.~that it contracts to a point in the thermodynamic limit, see \cite{HK}). We also give a result on the exponential decay of correlations in such impurity models. There, the decay rate is determined by the gap to the sector instead of the energy differences within the sector. Indeed, for disordered systems with many-body localization, one expects exponential decay of (most) correlations whereas the system is probably gapless. These applications are described in Section \ref{sec: applications}. 

Similar impurity models (restricting to perturbations of classical systems in the bulk) were studied before by Albanese \cite{A}. An interesting result of this flavour was given by Yarotsky on ground states that are perturbations of classical states. In \cite{Y} he showed that the influence of arbitrary boundary conditions, which in particular may also close the gap, decays exponentially into the bulk, underlining the importance of a `bulk gap' for decay of correlations besides the true spectral gap. 

Finally, in Section \ref{sec: Example}, we ask the natural question whether one can construct an exponentially local and unitary spectral flow, or whether there is some fundamental obstruction to the existence to such a flow. This question is particularly important in view of the wealth of applications of this technique, e.~g.~concerning (topological) ground state phases and their stability \cite{BH,BHM}, the quantum Hall effect \cite{HM}, or disordered systems \cite{H5}, and see also \cite{H6,H4,O}. In general, there seems to be no obvious way to obtain such a exponentially local unitary spectral flow, but we present a simple example where it can be obtained by exploiting the fact that in this instance the issue reduces to a few-particle problem.

\subsection*{Acknowledgements} It is a pleasure to thank Sven Bachmann for interesting discussions and help at the beginning of this project.
Furthermore, WDR and MS are thankful to the DFG (German Science Foundation) and the Belgian Interuniversity Attraction Pole (P07/18 Dygest) for financial support.

\section{Setup}

Let $\Gamma$ be the infinite volume, the set of vertices of an infinite graph equipped with the (shortest) graph distance metric $d$. Typically we think of the $\nu$-dimensional lattice $\Gamma=\bbZ^\nu$, but we will also need $\Gamma$  obtained by contracting a connected subset of vertices $S$ to a single vertex. We write $X\Subset \Gamma$ to indicate a finite subset of the infinite volume and if we speak of a volume $\Lambda$ we refer to a finite connected subset of $\Gamma$. 
For a given length $l\geq 0$, we call $X_l:=\{x\in\Lambda;
\, d(X,\{x\})\leq l\}$ the $l$-fattening of $X$. 

To each vertex $x\in\Gamma$ we assign a finite dimensional Hilbert space $\caH_x$ and to each subset $X\Subset \Gamma$ the tensor product $\caH_X:=\bigotimes_{x\in X} \caH_x$. The operators on $\caH_X$ endowed with the operator norm constitute the algebra $\caA_X$ of observables that, we say, have support in $X$. For $X\subset X'$ any $A\in\caA_X$ can be identified with a local operator $A\otimes \opunit$ acting on $\caH_{X'}$ and we will not distinguish between the two in notation. Note that we do not require all $\mathcal{H}_x$ to have the same dimension, which would be unnatural since we think of some vertices as corresponding to a region $S\subset \bbZ^\nu$.

An interaction $\Phi$ on $\Gamma$ is a family of self-adjoint local operators $\{\Phi(X)\}_{X\Subset \Gamma}$, $\Phi(X)\in\caA_X$. We say that it has a finite range $R>0$ if $\Phi(X)=0$ whenever $\mathrm{diam}(X)>R$ (diameter of $X$).  For each volume $\Lambda$ the interaction $\Phi$ defines a Hamiltonian (with open boundary conditions)
\begin{equation}
 H_\Lambda^{\Phi}:=\sum_{X\subset \Lambda} \Phi(X)
\end{equation}
and a one-parameter group of automorphisms $\{\tau_t^\Lambda\}$, $t\in\bbR$, on $\caA_\Lambda$ through
\begin{equation}
 \tau_t^{\Phi,\Lambda}(A)\equiv \tau_t^\Lambda(A):= \e^{i H_\Lambda^{\Phi} t} A \e^{-i H_\Lambda^{\Phi} t}, \quad A\in\caA_\Lambda
\end{equation}
which is the (Heisenberg) time evolution of the quantum spin system. The set of those sites within a subset $X\subset \Gamma$ which interact with sites of the complement $X^c$,
\begin{equation}
\begin{split}
  \partial^\Phi X :=\bigl\{ x\in X \,;\,\exists & Y\Subset \Gamma \text{ with } x\in Y \text{ and } \\
   &X\cap Y\neq \emptyset, \,X^c\cap Y\neq \emptyset, \,\Phi(Y)\neq 0 \bigr\}
\end{split}
\end{equation}
will be called the $\Phi$-boundary of $X$.

\subsection{Lieb--Robinson Bounds}\label{sec: LiebRobinson}
Here we briefly state a version of the Lieb--Robinson bounds on propagation in quantum spin systems as proven in \cite{NS07}. See also \cite{NS10} by the same authors, where it is more explicit that the Lieb--Robinson velocity (see $v$ below) can be defined independent of the single-site interaction. 
We fix a family of non non-increasing functions $F_\mu:[0,\infty)\rightarrow (0,\infty)$, $\mu\geq 0$, with $F_\mu(d):=\e^{-\mu d}F_0(d)$, such that, for $\mu=0$, and hence for all $\mu\geq0$, the following hold:\\
 (i) uniform summability
 \begin{equation}
  \lVert F_\mu \rVert := \sup_{x\in \Gamma} \sum_{y\in \Gamma} F_{\mu}\bigl(d(x,y)\bigr)<\infty
 \end{equation}
 (ii) convolution property
 \begin{equation}
  C_\mu := \sup_{x,y \in \Gamma} \sum_{z\in\Gamma} \frac{F_{\mu}\bigl(d(x,z)\bigr)F_{\mu}\bigl(d(z,y)\bigr)}{F_{\mu}\bigl(d(x,y)\bigr)} < \infty
 \end{equation}
To bring locality on the lattice into the game, we define a Banach space of exponentially decaying interactions $\mathscr{B}_\mu(\Gamma)$, consisting of $\Phi$ such that
\begin{equation}
 \lVert \Phi \rVert_\mu := \sup_{x,y \in \Gamma} \sum_{X \ni x,y} \frac{\lVert \Phi(X)\rVert}{F_{\mu}\bigl(d(x,y)\bigr)}<\infty
\end{equation}
The above expression without the contributions from single-site interaction terms, i.e., imposing additionally the restriction $\lvert X\rvert>1$ in the sum, is abbreviated with $\lVert \Phi \rVert_\mu'$. The following Theorem states the announced propagation bounds, which first appeared in \cite{LR}.
\begin{theorem}
 Let $\mu > 0$ and $\Phi \in \mathscr{B}_\mu(\Gamma)$. For every volume $\Lambda \Subset \Gamma$ and pair of local observables $A\in\caA_X$ and $B\in\caA_Y$ with supports $X,Y \subset \Lambda$, the time evolution $\tau_t^{\Phi,\Lambda}$ satisfies
 \begin{equation}
  \bigl\lVert [ \tau_t^{\Phi,\Lambda}(A), B ]\bigr\rVert\leq \frac{2\lVert F_0 \rVert}{C_\mu} \, \lVert A\rVert \lVert B\rVert  \min\bigl\{\lvert \partial^\Phi X \rvert, \lvert \partial^\Phi Y \rvert\bigr\} \e^{-\mu[ d(X,Y)-v\lvert t\rvert]}
 \end{equation}
for all times $t\in \bbR$ where 
\begin{equation}
 v:= \frac{2 \lVert \Phi \rVert_\mu' C_\mu}{\mu}
\end{equation}
is the so-called Lieb--Robinson velocity.
\end{theorem}
Note that the above estimate, and in particular the Lieb--Robinson velocity $v$, does not depend on the single-site interaction terms $\Phi(\{x\})$, $x \in \Gamma$, and also not on the volume $\Lambda$.

\subsection{Gapped Sectors and Local Perturbations}\label{sec: GappedSectors}

We would like to study the effect on eigenfunctions of the previously defined Hamiltonians $H_\Lambda^{\Phi}$ when adding local perturbations that have support only on a finite number of sites $k_1, k_2,\dots$ within a subset denoted by $K\Subset \Gamma$. For each of these sites let $W_{i}:[0,1]\rightarrow \caA_{\{k_i\}}$ be a smooth path of self-adjoint operators starting at $W_i(0)=0$.  We denote with $C_W\geq 0$ the maximal norm of the derivative of $W:=\sum_i W_i$. Then
\begin{equation}
 H_\Lambda(s):= H_\Lambda^{\Phi} + W(s)
\end{equation}
defines a smooth path of Hamiltonians. The restriction to (independent) perturbations at single sites is not severe since we are working with general infinite volume graphs and single-site Hilbert spaces. Generalization to perturbations involving several sites can be achieved by starting out with a modified infinite volume, in which these sites are merged to a single vertex. Note however that our estimates are sensitive to the norm of $W$ (e.g.\ via $C_W$) and hence in practice to the size of $K$.\\

Our \emph{main assumption} is that $H_\Lambda(s)$ maintains a sector in the spectrum which has a gap of at least $g>0$ to the rest of the spectrum, uniformly along the path and for all volumes $\Lambda$ under consideration (e.g.~for all volumes large enough to exclude non-interesting pathologies). Therefore the dependence on $\Lambda$ will be mostly suppressed in notation. More precisely, we assume as in \cite{BMNS} that the spectrum of $H(s)$ is a union
\begin{equation}
\sigma(s)=\sigma_{\mathrm{in}}(s)\cup \sigma_{\mathrm{out}}(s) 
\end{equation}
where $\sigma_{\mathrm{in}}(s)$ is contained in an interval of length at most $L\geq0$ which does not intersect $\sigma_{\mathrm{out}}(s)$. We require that the distance between these two parts of the spectrum is bounded below by $g>0$ uniformly for all $s\in[0,1]$ (and all volumes $\Lambda$).\\

The orthogonal projection onto the $D$-dimensional sector eigenspace is denoted by $P(s)$. For each $s\in[0,1]$, let $\lambda_i(s)$, $i=1,\dots,D$, represent the eigenvalues in the sector of $H(s)$ repeated according to their multiplicity and $\{\psi_i(s)\}_i$ a set of associated orthonormal eigenvectors $H(s)\psi_i(s)=\lambda_i(s)\psi_i(s)$. 
\\

Given the above setting, it was shown (\cite{BMNS}) that there exists a continuous differentiable path of unitaries $U(s)$ inducing the previously mentioned `spectral flow'
\begin{equation}\label{eq: SpectralFlow}
 P(s)=U(s) P(0) U(s)^*
\end{equation}
and solving the evolution equation  
\begin{equation}
 -i \partial_s U(s) = G(s)U(s), \qquad U(0)=\opunit
\end{equation}
with the generator $G(s)$  constructed explicitly as a self-adjoint quasi-local operator acting non-trivially only close to the perturbation. The accuracy with which $G(s)$ and hence also $U(s)$ can be approximated by a local operator in $\mathcal{A}_X$ is shown to increase `sub-exponentially' in $d(X^c,K)$. Again, our main motivation here is to complement these findings by a result with exponential accuracy. This comes at a cost: we are able to transform to the perturbed eigenvectors by means of quasi-local operators which decay exponentially in $d(X^c,K)$, but this transformation will not be unitary.

\section{Main Result}\label{sec: results}
\subsection{Weak Local Perturbation} \label{sec: WeakLocal}
In this paragraph we restrict to weak perturbations in the sense that we compare the eigenvectors in the sectors of $H(s_0+\varepsilon)$ and $H(s_0)$ for $s_0\in [0,1-\varepsilon]$ and $\varepsilon>0$ small enough. However, the following observations are in fact the main step in our analysis and the result on the influence of local perturbations along the path $H(s)$ will be obtained by iterating the argument. 

In what follows, we assume that an interaction $\Phi \in \mathscr{B}_{\mu}(\Gamma)$ is fixed, as described above, and that the path of Hamiltonians $H(s)$ satisfies the \emph{main assumption} stated in Section \ref{sec: GappedSectors}.  Therefore, constants appearing will be understood to depend only on the choice of $\mu$ and the parameters related to the interaction $\Phi$, i.e.\  $\lVert F_0 \rVert$, $C_\mu$, and $v$, and also on the parameters involved in the \emph{main assumption}:  the sector dimension $D$, the number of perturbation sites $\lvert K\rvert$ and $g$, $C_W$, and $L$.  Most importantly, constants are independent of the volume $\Lambda$.  The statement `for $\varepsilon$ small enough' and $O(\varepsilon)$ is to be understood similarly, the maximum value can depend on all these parameters but not the volume. 

By the existence of the gap $g>0$ and basic perturbation theory \cite{Ka}
\begin{equation}
\lVert P(s_0+\varepsilon)-P(s_0)\rVert = O(\epsilon) 
\end{equation}
and the vectors $\{P(s_0+\varepsilon)\psi_i(s_0)\}_{i=1,\dots, D}$ remain linearly independent, if $\varepsilon>0$ is small enough. Then the perturbed eigenvectors (normalized) can be written as
\begin{equation} \label{eq: SmallStep}
 \psi_i(s_0+\varepsilon)=\sum_{j} c_{ij}(s_0,\varepsilon) P(s_0+\varepsilon) \psi_j(s_0)
\end{equation}
for suitable coefficients $c_{ij}(s_0,\varepsilon)\in\bbC$, $i,j=1,\dots, D$. One rather obvious bound on these coefficients is
\begin{equation}\label{eq: BoundCoeff}
 \lVert c_i(s_0,\varepsilon) \rVert_1:= \sum_j \lvert c_{ij}(s_0,\varepsilon) \rvert \leq 2 \sqrt{D} 
\end{equation}
for $\varepsilon>0$ small enough. Finally we define 
\begin{equation}
 \xi:=\frac{g+4C_\mu \lVert \Phi \rVert_\mu'}{\mu g }=\frac{1}{\mu}+2 \frac{v}{g}
\end{equation}
which will be the relevant length scale for our locality estimates. The second equality shows that it is related to the natural length scales in the system.

\begin{proposition}\label{prop: Weak}
 For every $\mu'< 1/\xi$, there is $\varepsilon_0>0$ so that, for all $\varepsilon<\varepsilon_0$, $s_0\in[0,1-\varepsilon]$, $i=1,\dots,D$, and $l> 0$, there exists a local operator $R_i^l(s_0, \varepsilon)$ with support in the $l$-fattening $K_l$ and satisfying
 \begin{equation}\label{eq: NormBoundProp}
  \bigl\lVert  R_i^l(s_0, \varepsilon) \bigr\rVert \leq \lvert \sigma_{\text{in}}(s_0+\varepsilon)\rvert
 \end{equation}
  and
 \begin{equation}\label{eq: LocalApproxProp}
  \bigl\lVert \bigl(P(s_0+\varepsilon)- R_i^l(s_0, \varepsilon)\bigr) \psi_i(s_0)\bigr\rVert =  \e^{-\mu' l}  O(\varepsilon)
 \end{equation}
\end{proposition}

\begin{proof} In short, the proof relies on an idea from \cite{HMNS} to approximate spectral projection of gapped spectrum by a Gaussian integral of the time evolution operator, which then is accessible by perturbation theory and the Lieb--Robinson bounds.

First recall that changes of the interaction at single sites as from the perturbation $W(s)$ do not affect the interaction strength $\lVert \Phi \rVert_\mu'$ and in particular not the Lieb--Robinson velocity $v$. Therefore the following holds uniformly throughout the path, i.e., for every path parameter $s_0$. Here we view $V:= H(s_0+\varepsilon)-H(s_0)$ as small perturbation of $H(s_0)$, which is bounded by $\lVert V \rVert \leq \varepsilon C_W$. We also abbreviate $H_0\equiv H(s_0)$, $H\equiv H(s_0+\varepsilon)$ and the sector projections $P_0\equiv P(s_0)$ and $P\equiv P(s_0+\varepsilon)$. For any eigenvalue $\lambda\in\sigma\equiv \sigma(s_0+\varepsilon)$ and for any given parameter $\alpha>0$
and we introduce the operator
\begin{equation}\label{eq: ApproxProj}
 \caP_\lambda:=(\alpha/\pi )^{-1/2} \int_{\bbR}\d t \,\e^{- \alpha t^2} \, \e^{{it(H-\lambda)}}
\end{equation}
By inserting the spectral decomposition
\begin{equation}
 H-\lambda = \sum_{\kappa\in\sigma} (\kappa -\lambda) Q_\kappa
\end{equation}
where $Q_\kappa$ is the spectral projection for $H$ and eigenvalue $\kappa$, and noting that the Fourier transform of Gaussian is again Gaussian, we find that
\begin{equation}
 \caP_\lambda = \sum_{\kappa\in\sigma} \e^{- \frac{1}{4\alpha}(\kappa -\lambda)^2 } Q_\kappa
\end{equation}
Then there is a linear combination of these operators
\begin{equation}
 \caP:= \sum_{\lambda \in \sigma_{\mathrm{in}}} a_{\lambda} \caP_{\lambda}, \qquad 0<a_\lambda < 1 
\end{equation}
where the sum runs over all eigenvalues in the sector $\sigma_{\mathrm{in}}\equiv\sigma_{\mathrm{in}}(s_0+\varepsilon)$, that satisfies (we also set $\sigma_{\mathrm{out}}\equiv \sigma_{\mathrm{out}}(s_0+\varepsilon)$)
\begin{equation}
 \caP - P = \sum_{\kappa\in\sigma_{\mathrm{out}}} \sum_{\lambda\in\sigma_{\mathrm{in}}} \e^{-\frac{1}{4\alpha}(\kappa -\lambda)^2} Q_\kappa
\end{equation}
By spectral perturbation theory $\lVert Q_\kappa P_0 \rVert = O(\varepsilon)$ for $\kappa\in\sigma_{\mathrm{out}}$, and by the gap assumption we therefore obtain
\begin{equation}\label{eq: ApproxProjDiff}
 \bigl\lVert (\caP - P)P_0 \bigr\rVert = \e^{-\frac{1}{4\alpha}g^2}\cdot O( \varepsilon)
\end{equation}

In a next step each exponential in \eqref{eq: ApproxProj} is replaced by its Dyson--Phillips series, which is norm convergent in our finite dimensional setting, 
\begin{equation}
\begin{split}
 \e^{it(H-\lambda)}= \Bigl(\sum_{n=0}^\infty i^n\int_0^t \d t_1 \dots\int_0^{s_{n-1}}\d t_n \tau_{t_n}^0(V)\dots\tau_{t_1}^0(V) \Bigr) \e^{it(H_0-\lambda)}
\end{split}
\end{equation}
Here $\tau^0_t$ indicates the time evolution belonging to the Hamiltonian $H_0$ and the $(n=0)$ term in the sum is taken to be the identity. If $\psi_i^0$ is an eigenvector for an eigenvalue $\lambda_i^0$ in the sector of $H_0$, then we can rewrite
\begin{equation}
 \begin{split}
  \caP \psi_i^0&=\sum_{\lambda \in \sigma_{\mathrm{in}}}a_{\lambda}\,(\alpha/\pi )^{-1/2} \int_{\bbR}\d t \,\e^{- \alpha t^2} \, \e^{it(H-\lambda)}\,\e^{-it(H_0-\lambda)}\,\e^{it(\lambda_i^0-\lambda)}\; \psi_i^0 \\
  &=\sum_{\lambda \in \sigma_{\mathrm{in}}}a_{\lambda}\,(\alpha/\pi )^{-1/2} \int_{\bbR}\d t \, \e^{- \alpha t^2}\,\e^{it(\lambda_i^0-\lambda)}\\
  &\qquad\; \sum_{n=0}^\infty i^n\int_0^t \d t_1 \dots\int_0^{s_{n-1}}\d t_n \tau_{t_n}^0(V)\dots\tau_{t_1}^0(V) \; \psi_i^0\\
  &\equiv\bigl( R_i^{\geq T} + R_i^{\leq T}\bigr) \psi_i^0
 \end{split}
\end{equation}
where we also introduced a parameter $T\geq 0$ and where the operators $R_i^{\geq T}$ and $R_i^{\leq T}$ are defined by restricting the $t$-integration to $\lvert t\rvert\geq T$ and $\lvert t\rvert\leq T$ respectively. The $(n=0)$ term is exclusively added to $R_i^{\leq T}$. 

We now show that both terms in $R_i:=R_i^{\geq T} + R_i^{\leq T}$ can be well approximated by a local operator with support in $K_l$ if we choose
\begin{equation}
 \alpha:=\frac{g(g+4C_\mu \lVert \Phi \rVert_\mu')}{4 \mu l}\quad\text{and}\quad T:=\frac{4C_\mu \lVert \Phi \rVert_\mu' l}{(g+4C_\mu \lVert \Phi \rVert_\mu')v}
\end{equation}
giving the three equalities
\begin{equation}
 \frac{g^2}{4\alpha},\, \alpha T^2,\, \mu(l-vT) = \frac{\mu gl}{g+4C_\mu \lVert \Phi \rVert_\mu'}
\end{equation}
which will be the dominant exponents in the next estimates.

In fact, the effect of $R_i^{\geq T}$ is so small that its locality is irrelevant. A computation gives the following upper bound on its norm,
\begin{equation}
 \begin{split}
  2\sum_{\lambda \in \sigma_{\mathrm{in}}}a_{\lambda}(\alpha/\pi )^{-1/2} \int_{ t \geq T}\d t \, \e^{- \alpha t^2 }\,  t \lVert V \rVert\, \e^{ t \lVert V \rVert}= \e^{-\mu' l}\, O(\varepsilon)
 \end{split}
\end{equation}
for any 
\begin{equation}
 \mu'< \frac{\mu g }{g+4C_\mu \lVert \Phi \rVert_\mu'}
\end{equation}
and $\varepsilon $ small enough.
Concerning the second term $R_i^{\leq T}$, we can make use of the Lieb--Robinson bound at small times $\lvert t \rvert\leq T$ for the time evolution of the perturbation $\tau_t^0(V)$. Let $A\in \caA_X$ be any operator with support $X\subset \Lambda\setminus K_l$, i.e., at least a distance $l$ away from the perturbation in $K$, then its commutator with $R_i^{\leq T}$ is bounded by
\begin{equation}
\begin{split}
 \bigl\lVert\bigl[A,R_i^{\leq T}\bigr]\bigr\rVert&\leq \frac{2 \lVert F_0\rVert}{C_\mu}\lVert A\rVert \,\lvert K\rvert\,\e^{-\mu(l-vT)}\sum_{n=1}^\infty \frac{n}{n!} \bigl( \lVert V\rVert T\bigr)^n  \\
 &= \lVert A \rVert \,\e^{-\mu' l}\, O(\varepsilon)
 \end{split}
\end{equation}
for $\varepsilon$ small enough.
Therefore (see e.g.~\cite{BMNS,NSW}) the normalized partial trace
\begin{equation}
 R_i^{\leq T,l}:= \frac{1}{\mathrm{dim} \caH_{\Lambda\setminus K_l}} \mathrm{Tr}_{\Lambda\setminus K_l}\bigl( R_i^{\leq T}\bigr)
\end{equation}
which clearly has support in $K_l$, satisfies
\begin{equation}
 \bigl\lVert R_i^{\leq T} -R_i^{\leq T,l}\bigr\rVert =\e^{-\mu' l}\, O(\varepsilon)
\end{equation}
By the definition of the constants $\alpha$ and $T$, equation \eqref{eq: ApproxProjDiff} becomes
\begin{equation}
 \bigl\lVert (\caP - P)P_0 \bigr\rVert = \e^{-\mu' l}\, O( \varepsilon)
\end{equation}
Putting it all together we find that \eqref{eq: LocalApproxProp} in the proposition holds for $R_i^l(s_0,\varepsilon):=R_i^{\leq T,l}$ and that \eqref{eq: NormBoundProp} can be obtained from 
\begin{equation}
 \lVert R_i\rVert\leq \sum_{\lambda \in\sigma_{\text{in}}} a_\lambda\leq \lvert \sigma_{\text{in}}(s_0)\rvert
\end{equation}

\end{proof}

\subsection{Iterating the Argument }

A relation between eigenvectors $\psi_i(0)$ and $\psi_i(1)$ from the beginning and end of the path $H(s)$ can be obtained by iterating the step \eqref{eq: SmallStep} from the previous paragraph,
\begin{equation}
\begin{split}
  \psi_{i_n}(1)&=\sum_{i_0, \dots, i_{n-1}}\Bigl( \prod_{k=1}^n c_{i_k i_{k-1}}(k/n,1/n)\Bigr)P(1)\dots P(2/n)P(1/n)\psi_{i_0}(0)
\end{split}
\end{equation}
for $n\geq 1$ large enough (in the same sense as $\varepsilon$ small enough). In each term of the sum, $P(m/n)$ then can be substituted by $R_{i_{m-1}}^l((m-1)/2,1/n)$ as from the Proposition. Together with the bound \eqref{eq: BoundCoeff} on the coefficients $c_{ij}$, the following Theorem is an immediate consequence of the Proposition.  This is the result that motivates the title of this article. 
\begin{theorem} \label{thm: main}
 For a Hamiltonian path $H(s)$ as introduced above and for every $\mu'< 1/\xi$, there is a constant $C \geq 0$, such that, for any $l>0$ and $i=1,\dots,D$, there exist local operators $L_{ij}^l$, $j=1,\dots,D$, with support in the $l$-fattening $K_l$, which take eigenvectors from the sector of $H(0)$ to those of $H(1)$ according to
 \begin{equation}
  \bigl\lVert \psi_i(1)- \textstyle{\sum_j} L^l_{ij} \psi_j(0)  \bigr\rVert\leq C\,\e^{-\mu' l}
 \end{equation}
and which are uniformly bounded by $\lVert L_{ij}^l \rVert \leq C$. 

\end{theorem}

Note again that the constant $C$ may grow with the number of steps $n$ needed to follow the path in sufficiently small steps of size $\varepsilon=1/n$ (see the Proposition) and therefore may depend on all the model parameters.

\newpage

\section{Applications}\label{sec: applications}
{We present two examples where the above results give a useful insight. }

\subsection{Impurities in Systems with Unique Ground State}

Let $\Phi \in \mathscr{B}_\mu(\Gamma)$ be as above. However, we now assume additionally that each Hamiltonian $H^\Phi_\Lambda$ has a spectral gap $\gamma>0$ above a non-degenerate ground state energy. Let $\psi^{gs}$ be a normalized ground state vector, which is hence unique up to a phase factor. 
To each perturbation site $k\in K$ we attach a finite dimensional Hilbert space $\caI_k$ depicting the degrees of freedom of an impurity. In other words we redefine the single-site Hilbert space to become $\caH_k \otimes \caI_k$. The interaction $\Phi$ will be naturally viewed as family of operators on the larger space which act trivially on $\caI:= \bigotimes_k \caI_k$. The energy of each impurity and its coupling to the rest of the system is described by a self-adjoint operator $W_k$ on $\caH_k \otimes \caI_k$. To put ourselves in the context of Proposition \ref{prop: Weak}, we assume that $H\equiv H_\Lambda^{\Phi}$ can be connected to
\begin{equation}
 H':=H + W := H +\sum_k W_k
\end{equation}
by a smooth path of gapped Hamiltonians satisfying the \emph{main assumption} of Section \ref{sec: GappedSectors} with a $\dim(\caI)=D$ dimensional sector. In particular, everything holds uniformly in the volume $\Lambda$ which is therefore suppressed in the notation. By basic perturbation theory, the \emph{main assumption} is surely satisfied if the norm of $W$ is smaller than $\gamma$.   If all $\caI_k\cong \bbC$ are merely one dimensional Hilbert spaces, the overall Hilbert space and dimension of the ground state sector is not extended at all. Our results then effectively describe the influence of local perturbations on systems with unique gapped ground state, which was studied in detail in \cite{HMNS}. 

Let $\{\psi'_{i}\}$, $i=1,\dots,D$, be an orthonormal set of eigenvectors from the sector of $H'$. Let $P'$ be the orthogonal sector projection for $H'$ and $P$ the one for the sector of $H$. Note that the range of $P$ is the subspace of product vectors of the form $\psi^{gs}\otimes \phi$ with $\phi \in \caI$.

\begin{proposition}\label{prop: Impurity1}
For every $\mu'< 1/\xi$, there is a constant $C > 0$, such that for every length $l> 0$ and local observable $A$ with support outside of the $l$-fattening $K_l$
 \begin{equation}\label{eq: PropImpurity1}
  \bigl\lvert \bigl(\psi_i',A\psi_i'\bigr) - \bigl(\psi^{gs} ,A \psi^{gs}\bigr) \bigr\rvert \leq C \lvert \partial^\Phi K_{l/2}\rvert \lVert A\rVert \, \e^{-\mu' l}\qquad i=1,\dots, D
 \end{equation}
For every $l\geq 0$ there furthermore exists a local operator $T_l$ with support within $K_l$ which transforms the sector projection with exponential accuracy
\begin{equation}\label{eq: TrafoProj}
 \bigl\lVert P' -T_l^* P T_l^{} \bigr\rVert \leq C \, \e ^{-\mu' l}
\end{equation}
The norm of $T_l$ is bounded by $C$.
\end{proposition}

The Proposition can be proven by combining Theorem \ref{thm: main} with the exponential clustering Theorem for unique ground states, which we include here for the reader's convenience:

\begin{theorem}[\cite{NS07}]\label{thm: ExpCluster}
 Let $\mu>0$ and $\Phi\in \mathscr{B}_\mu(\Gamma)$ as before. Given $g>0$ there exist constants $\mu', c > 0$ such that, if $H_\Lambda^\Phi$ has a $g$-gapped ground state energy for a volume $\Lambda$ and if $\Omega$ is a normalized ground state vector of $H_\Lambda^\Phi$, then 
 \begin{equation}\label{eq: ExpCluster}
  \bigl\lvert (\Omega, AB \,\Omega )\bigr\rvert\leq c \lVert A\rVert \lVert B\rVert \min\bigl\{\lvert \partial^\Phi X \rvert, \lvert \partial^\Phi Y \rvert\bigr\}\e^{-\mu' d(X,Y)}
 \end{equation}
for all local observables $A,B$ with support in $X,Y$, which satisfy $P A \Omega=PA^*\Omega=0$ with $P$ being the spectral projection for the ground state subspace. The decay parameter can be chosen as 
 \begin{equation}
  \mu'=\frac{\mu g}{g+ 4 C_\mu\lVert \Phi\rVert_\mu'}
 \end{equation}
\end{theorem}
 
This type of Theorem dates back to \cite{H3,HK,NS}. In case that the ground state energy is a non-degenerate eigenvalue of $H_\Lambda^\Phi$ (as is relevant for the present proof), we can drop the condition on $A$ and instead of \eqref{eq: ExpCluster} we obtain
 \begin{equation}
  \bigl\lvert (\Omega ,AB\Omega )-(\Omega,A\Omega)(\Omega,B\Omega) \bigr\rvert\leq c \lVert A\rVert \lVert B\rVert \min\bigl\{\lvert \partial^\Phi X \rvert, \lvert \partial^\Phi Y \rvert\bigr\}\e^{-\mu' d(X,Y)}
 \end{equation}
 for all observables $A,B$ with support in $X,Y$.\\

The sector projection can be approximated as in \eqref{eq: TrafoProj} because of the simple product structure of vectors in the unperturbed sector, but this does not seem to be implied by the Theorem in general. However, recall once again that in general \eqref{eq: TrafoProj} holds with a sub-exponential bound by the spectral flow technique, see \eqref{eq: SpectralFlow}. In physically relevant models $\lvert \partial^\Phi K_{l/2}\rvert$ usually grows only polynomially in $l$, so that this factor in \eqref{eq: PropImpurity1} can be eliminated by choosing a slightly smaller decay parameter $\mu'$.

\begin{proof}[Proof of Proposition \ref{prop: Impurity1}]
First we choose an orthonormal basis $\{\psi_i\}$ in the ground state subspace of the unperturbed Hamiltonian $H$ of the form
$  \psi_i=\psi^{gs} \otimes \phi_i$, $i=1,\dots,D$,
with $\phi_i\in\caI$ and in Dirac's notation we define the local operators $I_{ij}=\lvert \phi_i \rangle \langle \phi_j \rvert$ acting non-trivially on $\caI$.
By Theorem \ref{thm: main} there exist, for any length $l\geq 0$, local operators $L_{ij}^l$ with support in $K_l$, so that the eigenvectors $\psi_i'$ in the system coupled to impurities satisfy
\begin{equation}\label{eq: Transform}
 \Bigl\lVert \psi_i' - \sum_j L_{ij}^l I_{ji}^{} \psi_i^{} \Bigr\rVert \leq \text{(const.)}\, \e^{-\mu' l}
\end{equation}
where here and in the following `(const.)' stands for a constant that may change from line to line, but may only depend on the fixed model parameters (not the volume). The greatest of these possibly different constants qualifies as the constant $C$ in the Proposition. $\mu'$ is of course taken from the Theorem. Therefore the perturbed sector projection can be approximated according to
\begin{equation}
 \Bigl\lVert P' - \sum_{i,p, q} L_{ip}^l I_{pi}^{} P I_{iq}^{} L_{iq}^{l*} \Bigr\rVert \leq \text{(const.)}\, \e^{-\mu' l}
\end{equation}
and we arrive at \eqref{eq: TrafoProj} with
\begin{equation}
 T_l:= \sum_{ij} L_{ij}^l I_{ji}^{}
\end{equation}
Concerning the first part of the Proposition note that
\begin{equation}\label{eq: TrafoState}
 \Bigl\lvert \bigl(\psi'_i,A \psi_i'\bigr)-\sum_{p ,q} \bigl( \psi_{m}^{},A I_{m p}^{} L_{i p}^{l/2*}  L_{i q}^{l/2} I_{ q m}^{} \psi_{m}^{}\bigr)\Bigr\rvert\leq \text{(const.)}\, \e^{-\mu' l/2}
\end{equation}
again with the aid of Theorem \ref{thm: main}. By the exponential clustering Theorem (Theorem \ref{thm: ExpCluster})
\begin{equation}
\begin{split}
  &\bigl\lvert \bigl( \psi_{m}^{},A I_{m p}^{} L_{i p}^{l/2*}  L_{i q}^{l/2} I_{ q m}^{} \psi_{m}^{}\bigr) - \bigl( \psi_{m}^{},A  \psi_{m}^{}\bigr)\bigl( \psi_{m}^{}, I_{m p }^{} L_{i p}^{l/2*}  L_{i q}^{l/2} I_{q m}^{} \psi_{m}^{}\bigr)\bigr\rvert\\
  &\leq \text{(const.)}\,\lvert \partial^\Phi K_{l/2} \rvert \lVert A\rVert\,\e^{-\mu' l}
\end{split}
\end{equation}
for any $m=1,\dots,D$ and a $\mu'>0$ as specified in the Proposition, which finishes the proof.

\end{proof}

\subsubsection{Exponential Decay of Correlations} 
Just above we stated and used the well-known fact that gapped unique ground states arising from local interactions exhibit exponential decay of spatial (truncated) correlations. The same was shown to hold for states defined through spectral projections on gapped sectors of eigenvalues which approach the ground state energy in the thermodynamic limit \cite{HK, NS06}. For our model with a single impurity we show a (modified) exponential clustering property for every state defined by an eigenvector from the low-lying sector. In particular it holds for any ground state (in finite volume) independent of the gap to the energy of the next highest excitation, which may be small but non-vanishing for all volumes. This is not surprising since we just showed that each sector state is identical to a gapped unique ground state away from the impurity and since, intuitively, one would expect the correlation length to be determined by bulk properties, e.g.~a kind of `bulk gap', of the system.  

For two subsets $X, Y \subset \Lambda$, which will be the supports of observables $A$ and $B$, and depending on the perturbation set $K$, we define an effective distance $d_K(X,Y)$ as the smallest number $l\geq 0$ such that $X_l\cup Y_l \cup K_l$ has a connected component containing both $X$ and $Y$. This distance hence tends to be shorter if an impurity is located in between the supports $X$ and $Y$. Our next result shows that correlations decay exponentially in this effective distance for a single impurity $K=\{k\}$ and we assume that $\lvert \partial^\Phi K_l \rvert$ increases at most polynomially in $l$. 


\begin{proposition}\label{prop: Impurity2}
 If $\lvert K\rvert =1$ and for every $\mu'<1/\xi$, there is a constant $C > 0$, so that, for every pair of observables $A,B$ with support in $X,Y$ 
 \begin{equation}
  \bigl\lvert \bigl(\psi'_i ,A B\psi_i'\bigr)-\bigl(\psi'_i ,A\psi_i'\bigr)\bigl(\psi'_i ,B\psi_i'\bigr)\bigr\rvert \leq C\,\lVert A \rVert \lVert B\rVert \e^{-\mu' d_{K}(X,Y)} 
 \end{equation}
 for $i=1,\dots,D$.
\end{proposition}

\begin{proof}
 If $d(X,Y)$ is smaller than or equal to the distance of $X$ and $Y$ to the perturbation set $K$, then exponential clustering for $\omega'_i$ is an immediate consequence of Proposition \ref{prop: Impurity1} knowing that $\omega^{gs}$ has this property, see Theorem \ref{thm: ExpCluster}. More generally, we use again that $\psi_i' \approx \sum_j L_{ij}^l I_{ji}^{} \psi_i^{}$ as in \eqref{eq: Transform} and that $L_{ij}^l$ commutes with either $A$ and/or $B$, let's say $B$, for $l=d_K(X,Y)/2$. In the same way that \eqref{eq: TrafoState} was obtained in the previous proof we arrive at
 \begin{equation}
 \Bigl\lvert \bigl(\psi'_i,A B \psi_i'\bigr)-\sum_{p ,q} \bigl( \psi_{m}^{}, I_{m p}^{} L_{i p}^{l*} A L_{i q}^{l} I_{ q m}^{}B \psi_{m}^{}\bigr)\Bigr\rvert\leq \text{(const.)}\, \e^{-\mu' l}
\end{equation}
Since the distance $d(X\cup K_l,Y)$ between the support of each $I_{m p}^{} L_{i p}^{l*} A L_{i q}^{l} I_{ q m}^{}$ and of $B$ is at least $d_K(X,Y)$ and since we assumed that $\lvert \partial^\Phi K_l\rvert$ grows slower than exponentially in $l$, the exponential clustering Theorem can be used to finish the proof.
 
\end{proof}

For \emph{several impurities} $\lvert K \rvert> 1$, the different regions around them may be correlated for arbitrary distances for states from (entangled) eigenvectors already due to possible degeneracies in the spectrum. This can of course already be seen in the uncoupled case with $W=0$. For the average of sector states $\omega(\cdot)=\frac{1}{D}\tr(P\,\cdot\,)$ however, truncated correlations still decay exponentially in the distance between observables as a consequence of the exponential clustering Theorem applied to the unique ground state in the bulk. If the impurities can be coupled to the system one after another along a gapped path, one might also expect exponential clustering to hold for $\omega'(\cdot)=\frac{1}{D}\tr(P'\,\cdot\,)$ as above in Proposition \ref{prop: Impurity2} (i.e.~with respect to distance $d_{K}$), but we are not able to show it. 

Also to clarify this statement, we briefly point out that in this case `sub-exponential' decay of correlations can be easily obtained with the help of the spectral flow technique, indicating the deficiency in Theorem \ref{thm: main} of not being implemented by a unitary transformation. The set of impurity sites can be split into two parts $K=K^X\cup K^Y$ of sites close to $X$ and close to $Y$ in the sense that, for $l=d_K(X,Y)/2$, their fattened sets satisfy 
\begin{equation}
 d\bigl(K_l^X\cup X, K_l^Y\cup Y\bigr)\geq l
\end{equation}
Furthermore we assume that both the uncoupled and perturbed Hamiltonian $H$ and $H'$ are connected to 
\begin{equation}
H^X:=H+\sum_{k\in K^X}W_k 
\end{equation}
by a gapped path in the same way as before, but with derivatives in $\caA_{K^X}$ and $\caA_{K^Y}$ respectively. This requirement formalizes the assumption that the impurities can be coupled to the system one after another without closing the gap. The spectral flow technique \cite{BMNS} allows to construct two unitary operators $U^X_{l}$ and $U^Y_{l}$ supported on $K_l^X$ and $K_l^Y$ which transform between the gapped sectors of $H$/$H^X$ and $H^X$/$H'$,
\begin{equation}
 U^X_{l} P (U^X_{l})^* \approx P_{}^X \quad \text{ and }\quad U^Y_{l} P_{}^X (U^Y_{l})^* \approx P'
\end{equation}
up to errors whose norm decays sub-exponentially in $l$ indicated by the $\approx$ symbol. Here, $P^X$ stands for the sector projection of $H^X$. Then indeed
\begin{equation}
 \begin{split}
  \omega'\bigl(AB\bigr)&\approx \omega \bigl( (U^X_l)^* (U^Y_l)^* A   B U_l^Y U_l^X\bigr) \\
  &= \omega \bigl( (U^X_l)^* A U_l^X (U^Y_l)^* B U_l^Y\bigr) \\
  &\approx \omega \bigl( (U^X_l)^* A U_l^X \bigr)\omega \bigl( (U^Y_l)^* B U_l^Y\bigr)\\
  &=\omega \bigl( (U^X_l)^*(U^Y_l)^* A  U_l^Y U_l^X \bigr)\omega \bigl((U^X_l)^* (U^Y_l)^* B U_l^Y U_l^X\bigr)\\
  &\approx \omega'\bigl( A\bigr) \omega'\bigl(B\bigr)
 \end{split}
\end{equation}
To get to the third line, the exponential clustering property was used for the state $\omega$.

\subsection{Impurities in Systems with Topological Quantum Order}

Another interesting class of models for which Theorem \ref{thm: main} translates into exponentially sharp locality estimates are those with topological quantum order (defined below). 
To account for the topological aspects of such systems, each volume $\Lambda$ here is a finite graph imbedded on a possibly non-trivial surface rather than a finite subgraph of a common infinite volume $\Gamma$. 
For the sake of concreteness, we restrict to two dimensional square lattices $\Lambda=\bbZ_L \times \bbZ_L$ of length $L$ with periodic boundary conditions (imbeddings on a torus) and to translation invariant finite-range interactions. This restriction includes one of the most prominent examples with topological order, Kitaev's toric code, see \cite{Ki}. To relate the models at different $L$, we fix a $\mu$ and we assume that the model parameters $\lVert F_0 \rVert$, $C_\mu$, and $v$ converge to finite non-zero values as $L\to \infty$. This puts us in a setting where  Theorem \ref{thm: main} applies ($\Lambda$-uniformly).   

We are concerned with impurity models which are set up for each such $\Lambda$ in the same way as in the previous section and the same notation and assumptions will be used. The only change we make is that we dismiss the requirement of a unique ground state for $H$ and instead we assume, for each volume $\Lambda$, the ground states of $H$ possess topological quantum order {\small (TQO)} defined as in \cite{BH,BHM,BHV}, see also \cite{MZ}: \\
\emph{For any local observable $A$ whose support $X$ is contained within a square of side length $L^*$ there is $z\in\bbC$, so that $PAP=zP$.}\\
This implies that different ground states cannot be distinguished by measurements localized on length scales up to $L^*$. Typically we have in mind that $L^*$ increases as a function of the system's size $L$, e.~g.~that $L^*\geq L^a$ for some $a>0$. 

\begin{proposition}\label{prop. Topo}
 For every $\mu'<1/\xi$, there is a constant $C > 0$, such that the following holds for every length $l\geq 0$: If $A$ is a local observable with support $X$ outside the $l$-fattening $K_l$ of the impurity set $K$ such that the union $X\cup K_{l/2}$ is contained in a square of side length $L^*$, then
 \begin{equation}
  \bigl\lvert \bigl(\psi_i', A \psi_j'\bigr) -z\delta_{ij} \bigr\rvert \leq C \lVert A\rVert \, \e^{-\mu' l}
 \end{equation}
 for $i,j=1,\dots,D$ and where $z\in\bbC$ is given by $PAP=zP$.
\end{proposition}
The Proposition shows that away from the impurities in the system (or from perturbations if all $\caI_k\cong \bbC$) the measurements of local (up to $L^*$) observables remain unchanged with exponential accuracy. 
Moreover, in this region {\small TQO} persist `with error' that is exponentially decreasing. This terminology was coined in \cite{BHV}. 

\begin{proof}
 We essentially follow the proof for Proposition \ref{prop: Impurity1} with adaptations due to the possible, say $f$-fold, degeneracy of the ground state energy for the bulk system (without the impurities). Let $\{\psi_i\}$ be an orthonormal basis in the ground state subspace for $H$ of the form 
 \begin{equation}
 \begin{split}
  &\psi_i = \psi^{gs}_{i_b} \otimes \phi_{i_K}, \qquad i\equiv (i_b, i_K) \\ 
  &i_b=1,\dots, f,\quad i_K=1,\dots,\dim(\caI)
 \end{split}
 \end{equation}
 Note that $D=f\cdot \dim(\caI)$. By using Theorem \ref{thm: main} and the local transition operators $I_{i_K,j_K}$ as before we find
\begin{equation}
\begin{split}
 &\Bigl\lvert \bigl(\psi'_i,A \psi_j'\bigr)-\sum_{p,q} \bigl( \psi_{(p_b,m_K)}^{},A I_{m_K p_K}^{} L_{i p}^{l/2*}  L_{j q}^{l/2} I_{ q_K m_K}^{} \psi_{(q_b, m_K)}^{}\bigr)\Bigr\rvert\\
 &\leq \text{(const.)}\, \e^{-\mu' l/2}
\end{split}
\end{equation}
for any $m_K=1,\dots,\dim(\caI)$. By the definition of the local transformation operators and assumptions of the Proposition, the support of 
\[A I_{m_K p_K}^{} L_{i p}^{l/2*}  L_{j q}^{l/2} I_{ q_K m_K}^{}\]
lies within a square of length $L^*$ and therefore those terms in the sum with $p_b\neq q_b$ vanish as a consequence of the {\small TQO} condition. We can assume that the support $X$ of $A$ is disjoint from $K$ otherwise there is nothing to be proven. On the square lattice $\lvert \partial^\Phi K_{l/2}\rvert$ only grows polynomially in $l$. Again by the {\small TQO} condition
\begin{equation}
 \bigl( \psi_{(p_b,m_K)}^{},A  \psi_{(q_b, m_K)}^{}\bigr)=z\delta_{p_b q_b}
\end{equation}
and in particular
\begin{equation}
 P (A-z) \psi_{(p_b,m_K)}=P (A^*-\bar{z})\psi_{(p_b,m_K)}=0
\end{equation}
for all $p_b$ and $m_K$, so that the proof of the Proposition follows by the exponential clustering Theorem \ref{thm: ExpCluster}.

\end{proof}

\section{Simple Example of an Exponentially Local Spectral Flow}\label{sec: Example}

Here we present a particular impurity model in the same setup as above for which it is not difficult to construct a unitary exponentially local spectral flow. We are dealing with a $\nu$-dimensional lattice of $S=1/2$ spins, i.e.~$\caH_x\cong \bbC^2$, and restrict to cubic volumes $\Lambda=\bbZ_L^\nu$ of length $L$ with periodic boundary conditions. The bulk interaction is that of an `xy-model' with Hamiltonian
\begin{equation}
 H=-\sum_{d( x, y) =1} \bigl( S_x^1  S_y^1 +S_x^2 S_y^2 \bigr) +\sum_{x} (u(x)+2\nu) \bigl( 1/2 +S_x^3\bigr)
\end{equation}
where $S_x^i$, $i=1,2,3$, are the Pauli matrices at site $x$. For $u$ we take a positive function bounded below by $u(x)\geq  \gamma$, for $\gamma>0$, which then is also the spectral gap above the zero ground state energy. The ground state is the product of `spin-down' vectors (eigenvectors of $S_x^3$ and eigenvalue $-1/2$) and frustration-free. Note that $(1/2 +S_x^3)$ is the orthogonal projection onto the `spin-up' state and that in the first sum we can replace $ S_x^1  S_y^1 +S_x^2 S_y^2$ by $S_x^+ S_y^-$ in terms of the spin raising/lowering operators $S_x^\pm = S_x^1 \pm i S_x^2$. In fact, this model is unitarily equivalent to a system of hard core bosons where sites occupied by a boson correspond to those being in spin-up state (sometimes called Matsubara--Matsueda correspondence \cite{MM}). Locality of operators is strictly preserved in this correspondence. Using standard second quantization notation, see e.g.~\cite{DG}, 
\begin{equation}
 H\cong \widetilde{H}:= P_{\text{hc}}\, \mathrm{d}\Gamma(-\Delta +u)\, P_{\text{hc}}
\end{equation}
where $P_{\mathrm{hc}}$ is the orthogonal projection on the subspace of bosonic Fock space $\Gamma(l^2(\Lambda)\bigr)$ with at most one Boson per site (hard core condition). $\Delta$ denotes the discrete Laplacian. The particle number is conserved in this many-boson system just as our $xy$-model conserves the spin, i.e.~$[H,\sum_x S_x^3]=0$.\\

We now add an impurity at a single site, $K=\{k\}$, which itself consists of $N$ spins with Hilbert space $\caI=(\bbC^2)^{\otimes N}$. This is equivalent to adding $N$ sites $I=\{i_1,\dots,i_N\}$ on the Boson model side of the correspondence. The coupling to the system is described by a smooth path $W(s)$ of spin-conserving operators on $\caH_k\otimes \caI$ with $H(0)=H$, so that $H(s)=H+W(s)$ maintains a $2^N$-dimensional $g$-gapped sector. Recall that we denoted with $C_W$ the maximal norm of the derivative of $W$. Initially, any state with Bosons located only at the Impurity's sites is a ground state. By the assumption of spin/particle number conservation along the path, the low-lying sector is fully described by eigenfunctions of a few-particle system (up to $N$). These decay exponentially, making the model so tractable.

\begin{proposition}
 Let $P(s)$ be the spectral sector projections for the low-lying sector along the path of Hamiltonians $H(s)$ in the $xy$-model with impurity as defined above. There are $C,\mu'>0$ depending only on $g$ and $N$, such that the following holds: for every length $l\in[0,\infty]$, there is a smooth path 
 of self-adjoint operators $G_l(s)$ and of unitary operators $U_l(s)$ with support in the $l$-fattening $K_l$ of the impurity site, which satisfy
 \begin{equation}\label{eq: SpecFlow}
  \bigl\lVert P(s) - U_l^{}(s) P(0) U_l^{*}(s) \bigr\rVert \leq C \,\e^{-\mu'l}
 \end{equation}
and solve
\begin{equation}\label{eq: EvoEq}
 -i\partial_s U_l^{}(s)=G_l(s)U_l^{}(s), \quad U_l^{}(0)=\textnormal{\opunit}
\end{equation}
\end{proposition}

We only give a brief \emph{sketch of the proof}: 
We show that Kato's `transformation function', see II.~\S~4 in \cite{Ka}, which is a particular choice of a spectral flow, is exponentially quasi-local for our model. It is defined as the unique solution $U(s)$ of \eqref{eq: EvoEq} generated by the commutator
\begin{equation}
 G(s):= i \bigl[P(s),\partial_s P(s)\bigr]
\end{equation}
and satisfies \eqref{eq: SpecFlow} with $l=\infty$. For $n\geq 0$, let $H^{(n)}(s)$ and $P^{(n)}(s)$ be the restrictions of $H(s)$ and $P(s)$ to the invariant $n$-spin (particle) subspace. If $n>N$, note that $P^{(n)}(s)=0$, which implies $G^{(n)}(s)=0$ and $U^{(n)}(s)=\opunit$ for the restrictions of $G(s)$ and $U(s)$. By the holomorphic functional calculus, the sector projection and its derivative can be expressed in terms of the resolvent as
\begin{equation}
 \begin{split}
  &P^{(n)}(s)=-\frac{1}{2 \pi i} \int_{\caC(s)} \mathrm{d}z\, \bigl( H^{(n)}(s)-z \bigr)^{-1}\\
  &\partial_s P^{(n)}(s)=\frac{1}{2 \pi i} \int_{\caC(s)} \mathrm{d}z\,\bigl( H^{(n)}(s)-z \bigr)^{-1}\, \partial_s W^{(n)}(s) \,\bigl( H^{(n)}(s)-z \bigr)^{-1}
 \end{split}
\end{equation}
for a family of contours $\caC(s)$ at distance of at least $g/2$ to the spectrum and enclosing the low-lying spectrum. Switching to the particle description of our model, $H^{(n)}(s)$ is unitarily equivalent to to an $n$-particle discrete Schr\"odinger operator on (the symmetric hard-core subspace of) $l^2((\Lambda\cup I)^n)$. For such operators, a Combes--Thomas type estimate \cite{CT} shows that the matrix elements of the resolvent as in the above integrals decay exponentially away from the diagonal in the canonical position basis. The length scale of this decay may increase with the dimension $n$ and as the gap closes (for $z\in\caC(s)$). Since $\partial_s P^{(n)}(s)$ is equivalent to an operator with support in $K\cup I$, the matrix elements of $G^{(n)}(s)$ then decay exponentially away from $k$. Truncating $G^{(n)}(s)$ outside of $K_l$ only gives a difference in norm that decays exponentially in $l$. Therefore, we can define $G_l(s)$ of the Proposition as the direct sum (from $0$ to $N$) of these 
truncations and $U_l(s)$ 
as the corresponding unique solution of differential equation \eqref{eq: EvoEq}.

\vspace{11mm}
\bibliographystyle{plain}

\begin{thebibliography}{10}
 
 \bibitem{A} C.~Albanese, \emph{Unitary dressing transformations and exponential decay below threshold for quantum spin systems I-IV}. Comm.~Math.~Phys.~\textbf{134}, 1--27 and 237--272 (1990) 
 
 \bibitem{BMNS} S.~Bachmann, S.~Michalakis, B.~Nachtergaele, and R.~Sims, \emph{Automorphic Equivalence within Gapped Phases of Quantum Lattice Systems}. Comm.~Math.~Phys.~\textbf{309}, 835--871 (2012)
 
 \bibitem{BH} S.~Bravyi and M.~B.~Hastings, \emph{A short proof of stability of topological order under local perturbations}. Comm.~Math.~Phys.~\textbf{307}, 609--627 (2011)
 
 \bibitem{BHM} S.~Bravyi, M.~B.~Hastings, and S.~Michalakis, \emph{Topological quantum order: Stability under local perturbations}. J.~Math.~Phys.~\textbf{51}, 093512 (2010)
 
 \bibitem{BHV} S.~Bravyi, M.~B.~Hastings, and F.~Verstraete \emph{Lieb-Robinson Bounds and the Generation of Correlations and Topological Quantum Order}. Phys.~Rev.~Lett.~\textbf{97}, 050401 (2006)
 
 \bibitem{CT} J.~Combes and L.~Thomas, \emph{Asymptotic behaviour of eigenfunctions for multiparticle Schrödinger operators}. Comm.~Math.~Phys.~\textbf{34},  251--270 (1973)
 
 
 \bibitem{DG} J.~Derezi\'nski and C.~G\'erard, \emph{Mathematics of quantization and quantum fields}. Cambridge University press (2013)

 
 \bibitem{HMNS} E.~Hamza, S.~Michalakis, B.~Nachtergaele, and R.~Sims, \emph{Approximating the ground state of gapped quantum spin systems}. J.~Math.~Phys.~\textbf{50}, 095213 (2009)
 
 \bibitem{Ka} T.~Kato, \emph{Perturbation Theory for Linear Operators}. Springer (1980)
 
 \bibitem{Ki} A.~Kitaev, \emph{Fault-tolerant quantum computation by anyons}. Ann.~Phys.~\textbf{303}, 2--30 (2003)
 
 \bibitem{H} M.~B.~Hastings, \emph{Lieb--Schultz--Mattis in higher dimensions}. Phys.~Rev.~B \textbf{69}, 104431 (2004)
 
 \bibitem{H3} M.~B.~Hastings, \emph{Locality in quantum and Markov dynamics on lattices and networks}. Phys.~Rev.~Lett.~\textbf{93}, 140402 (2004) 
 
 \bibitem{H6}  M.~B.~Hastings, \emph{Quasi-adiabatic continuation in gapped spin and fermion systems: Goldstone's theorem and flux periodicity}. J.~Stat.~Mech.~05010 (2007)
 
 \bibitem{H4} M.~B.~Hastings, \emph{Locality in Quantum Systems}. In: Lecture Notes of the Les Houches Summer School \textbf{95}. \emph{Quantum Theory from Small to Large Scales}, OUP Oxford (2010)
 
 \bibitem{H5} M.~B.~Hastings, \emph{Quasi-adiabatic Continuation for Disordered Systems: Applications to Correlations, Lieb-Schultz-Mattis, and Hall Conductance}.  arXiv:1001.5280 [math-ph] (2010)
 
 \bibitem{HK} M.~B.~Hastings and T.~Koma, \emph{Spectral Gap and Exponential Decay of Correlations}. Comm.~Math.~Phys.~\textbf{265}, 781--804 (2006)
 
 \bibitem{HM} M.~B.~Hastings and s.~Michalakis, \emph{Quantization of Hall Conductance For Interacting Electrons on a Torus}. Comm.~Math.~Phys.~\textbf{334}, 433--471 (2015)
 
 \bibitem{HW}  M.~B.~Hastings and X.~Wen, \emph{Quasi-adiabatic continuation of quantum states: the stability of topological ground state degeneracy and emergent gauge invariance}. Phys.~Rev.~B \textbf{72}, 045141 (2005)
 
 \bibitem{LR} E.~Lieb and D.~Robinson, \emph{The Finite Group Velocity of Quantum Spin Systems}. Comm.~Math.~Phys.~\textbf{28}, 251--257 (1972)
 
 \bibitem{MM} T.~Matsubara and H.~Matsuda, \emph{A lattice model of liquid helium, I}. Progr.~Theor.~Phys.~\textbf{16}, 569--582 (1956)

 
 \bibitem{NS} B.~Nachtergaele and R.~Sims, \emph{Much ado about something -- Why Lieb--Robinson bounds are useful}. IAMP News Bull.~\textbf{Oct.~2010}, 22--29
 
 \bibitem{NSW} B.~Nachtergaele, V.~B.~Scholz, and R.~F.~Werner, \emph{Local Approximation of Observables and Commutator Bounds}. Op.~Th.: Adv.~Appl.~\textbf{227}, 143--149 (2013)
 
 \bibitem{NS06} B.~Nachtergaele and R.~Sims, \emph{Lieb--Robinson Bounds and the Exponential Clustering Theorem}. Comm.~Math.~Phys.~\textbf{265}, 119--130 (2006)
 
 \bibitem{NS07} B.~Nachtergaele and R.~Sims, \emph{Locality Estimates for Quantum Spin Systems}. In: New Trends in Mathematical Physics. \emph{Selected contributions of the XVth International Congress on Mathematical Physics}, 591--614, Springer (2009)
 
 \bibitem{NS10} B.~Nachtergaele and R.~Sims, \emph{Lieb-Robinson Bounds in Quantum Many-Body Physics}. In \emph{Entropy and the Quantum}, R.~Sims and D.~Ueltschi (Eds), Contemporary Mathematics \textbf{529}, American Mathematical Society, pp 141-176 (2010)
 
 \bibitem{MZ} S.~Michalakis and J.~Zwolak, \emph{Stability of Frustration-Free Hamiltonians}. Comm.~Math.~Phys.~\textbf{322}, 277--302 (2013)
 
 \bibitem{O} T.~J.~Osborne, \emph{Simulating adiabatic evolution of gapped spin systems}. Phys.~Rev.~A \textbf{75}, 032321 (2007)
 
 \bibitem{Y} D.~A.~Yarotsky, \emph{Uniqueness of the Ground State in Weak Perturbations of Non-Interacting Gapped Quantum Lattice Systems}. J.~Stat.~Phys.~\textbf{118}, 119--144 (2005)
\end{thebibliography}

\end{document}